\documentclass[letterpaper, 10 pt, conference]{ieeeconf}
\IEEEoverridecommandlockouts
\usepackage{graphics}
\usepackage{epsfig}
\usepackage{amsmath}
\usepackage{amssymb}
\usepackage{xcolor}
\let\labelindent\relax
\usepackage{enumitem}
\usepackage{multirow}
\usepackage{url}

\usepackage{amsthm}
\theoremstyle{definition}

\newtheorem{theorem}{\normalfont\bfseries Theorem}

\newtheorem{definition}{\normalfont\bfseries Definition}

\newtheorem{remark}{\normalfont\bfseries Remark}
\newtheorem{example}{\normalfont\bfseries Example}

\newcommand{\R}{\mathbb{R}}
\newcommand{\X}{\mathbb{R}^n}
\newcommand{\U}{\mathbb{R}^m}
\newcommand{\C}{\mathcal{C}}
\newcommand{\K}{\mathcal{K}}
\newcommand{\Ke}{\K^{\rm e}}
\newcommand{\Kinf}{\K_\infty}
\newcommand{\Keinf}{\K_\infty^{\rm e}}
\newcommand{\grad}[1]{\nabla #1}
\renewcommand{\L}[2]{L_{#1} #2}
\newcommand{\kd}{k_{\rm d}}

\newcommand{\sat}{{\rm sat}}
\newcommand{\Kp}{K_{\rm p}}
\newcommand{\xd}{x_{\rm d}}
\newcommand{\umax}{u_{\rm max}}
\newcommand{\Cc}{\C_{\rm c}}
\newcommand{\hc}{h_{\rm c}}
\newcommand{\Lor}{L_{\cup}}
\newcommand{\Land}{L_{\cap}}

\title{\LARGE \bf
Composing Control Barrier Functions for Complex Safety Specifications
}

\author{Tamas G. Molnar and Aaron D. Ames%
\thanks{*This research is supported in part by the National Science Foundation (CPS Award \#1932091) and Nodein Inc.}%
\thanks{T. G. Molnar and A. D. Ames are with the Department of Mechanical and Civil Engineering, California Institute of Technology, Pasadena, CA 91125, USA,
{\tt\small tmolnar@caltech.edu, ames@caltech.edu}.}%
\vspace{-1mm}
}

\begin{document}

\maketitle
\thispagestyle{empty}
\pagestyle{empty}

\begin{abstract}
The increasing complexity of control systems necessitates control laws that guarantee safety w.r.t.~complex combinations of constraints.
In this letter, we propose a framework to describe compositional safety specifications with control barrier functions (CBFs).
The specifications are formulated as Boolean compositions of state constraints, and we propose an algorithmic way to create a single continuously differentiable CBF that captures these constraints and enables safety-critical control.
We describe the properties of the proposed CBF, and we demonstrate its efficacy by numerical simulations.

\end{abstract}

\section{INTRODUCTION}
\label{sec:intro}

Control designs with formal safety guarantees have long been of interest in engineering.
Safety is often captured as constraints on the system's states that must be enforced for all time by the controller.
To enable the satisfaction of state constraints with formal guarantees of safety, control barrier functions (CBFs)~\cite{AmesXuGriTab2017} have become a popular tool in nonlinear control design.
As the complexity of safety-critical control systems increases, complex combinations of multiple safety constraints tend to arise, which creates a need for controllers incorporating multiple CBFs.

The literature contains an abundance of studies on multiple safety constraints.
Some approaches directly used multiple CBFs in control design.
For example,~\cite{rauscher2016constrained, xu2018constrained} directly imposed multiple CBF constraints on the control input in optimization-based controllers;
\cite{shawcortez2022robust} synthesized controllers by switching between multiple CBFs whose superlevel set boundaries do not intersect;
\cite{tan2022compatibility} investigated the compatibility of CBFs;
\cite{breeden2023compositions} ensured feasible controllers with multiple CBFs;
and~\cite{liu2020blf, liu2021adaptive} addressed multi-objective constraints via barrier Lyapunov functions.
These works usually linked safety constraints with AND logic: they maintained safety w.r.t.~constraint~1 AND constraint~2, etc.
Other approaches combined multiple constraints into a single CBF.
These include versatile combinations, such as Boolean logic with both AND, OR and negation operations, which was established in
\cite{glotfelter2017nonsmooth, Glotfelter2021} by nonsmooth barrier functions.
Similarly,~\cite{wang2016multiobjective} used Boolean logic to create a smooth CBF restricted to a safe set in the state space;
\cite{mitchell2022adaptation} combined CBFs with AND logic via parameter adaptation;
while~\cite{lindemann2019stl, lindemann2019conflicting} used signal temporal logic to combine CBFs in a smooth manner.

In this letter, we propose a framework to capture complex safety specifications by CBFs.
We combine multiple safety constraints via Boolean logic, and propose an algorithmic way to establish a single CBF for nontrivial safety specifications.
Our method leverages both the Boolean logic from~\cite{glotfelter2017nonsmooth} and the smooth combination idea from~\cite{lindemann2019stl}, while merging the benefits of these approaches.
We address multiple levels of logical compositions of safety constraints, i.e., arbitrary combinations of AND and OR logic, which was not established in~\cite{lindemann2019stl}, while we create a continuously differentiable CBF to avoid discontinuous systems like in~\cite{glotfelter2017nonsmooth}.
Meanwhile, as opposed to~\cite{wang2016multiobjective}, the stability of the safe set is guaranteed.

In Section~\ref{sec:CBF}, we introduce CBFs and motivate multiple safety constraints.
In Section~\ref{sec:multipleCBFs}, we propose a single CBF candidate to address the compositions of multiple constraints.
We also characterize its properties, and we use simulations to demonstrate its ability to address safety-critical control with nontrivial constraints.
Section~\ref{sec:concl} closes with conclusions.


\section{CONTROL BARRIER FUNCTIONS}
\label{sec:CBF}

We consider affine control systems with state ${x \in \X}$, control input ${u \in \U}$, and dynamics:
\begin{equation}
    \dot{x} = f(x) + g(x) u,
\label{eq:system}
\end{equation}
where ${f: \X \to \X}$ and ${g: \X \to \R^{n \times m}}$ are locally Lipschitz.
Our goal is to design a controller ${k: \X \to \U}$, ${u = k(x)}$ such that the closed-loop system:
\begin{equation}
    \dot{x} = f(x) + g(x) k(x),
\label{eq:closedloop}
\end{equation}
satisfies certain safety specifications.

If $k$ is locally Lipschitz, then for any initial condition ${x(0) = x_0 \in \X}$ system~\eqref{eq:closedloop} has a unique solution ${x(t)}$, which we assume to exist for all ${t \geq 0}$.
We say that the system is safe if the solution $x(t)$ evolves inside a {\em safe set} $\C$.
Specifically, we call~\eqref{eq:closedloop} {\em safe w.r.t.~$\C$} if ${x_0 \in \C \implies x(t) \in \C}$ ${\forall t \geq 0}$.
We define the safe set as the 0-superlevel set of a continuously differentiable function ${h: \X \to \R}$:
\begin{equation}
    \C = \{x \in \X: h(x) \geq 0 \},
\label{eq:safeset}
\end{equation}
assuming it is non-empty and has no isolated points.
Later we extend this definition to more complex safety specifications.

The input $u$ affects safety through the derivative of $h$:
\begin{equation}
    \dot{h}(x,u) =
    \underbrace{\grad{h}(x) f(x)}_{\L{f}{h(x)}}
    + \underbrace{\grad{h}(x) g(x)}_{\L{g}{h(x)}} u,
\label{eq:hdot}
\end{equation}
where
$\L{f}{h}$ and $\L{g}{h}$ are the Lie derivatives of $h$ along $f$ and $g$.
By leveraging this relationship, {\em control barrier functions (CBFs)}~\cite{AmesXuGriTab2017} provide controllers with formal safety guarantees.

\begin{definition}[\cite{AmesXuGriTab2017}]
Function $h$ is a {\em control barrier function} for~\eqref{eq:system} on $\X$ if there exists ${\alpha \!\in\! \Keinf}$\footnote{Function ${\alpha : (-b,a) \to \R}$, ${a,b>0}$ is of extended class-$\K$ (${\alpha \in \Ke}$) if it is continuous, strictly increasing and ${\alpha(0)=0}$.
Function ${\alpha : \R \to \R}$ is of extended class-$\Kinf$ (${\alpha \in \Keinf}$) if ${\alpha \in \Ke}$ and ${\lim_{r \to \pm \infty} \alpha(r) = \pm \infty}$.} such that for all ${x \!\in\! \X}$:
\begin{equation}
    \sup_{u \in \U} \dot{h}(x,u) \geq - \alpha \big( h(x) \big).
\label{eq:CBF_condition}
\end{equation}
\end{definition}

\noindent Note that the left-hand side of~\eqref{eq:CBF_condition} is $\L{f}{h(x)}$ if ${\L{g}{h(x)} = 0}$ and it is $\infty$ otherwise.
Thus,~\eqref{eq:CBF_condition} is equivalent to\footnote{In \eqref{eq:CBF_condition}-\eqref{eq:CBF_condition_rewritten}, strict inequality ($>$) can also be required rather than non-strict inequality ($\geq$) to ensure the continuity of the underlying controllers~\cite{jankovic2018robust}.}:
\begin{equation}
    \L{g}{h(x)} = 0 \implies \L{f}{h(x)} + \alpha \big( h(x) \big) \geq 0.
\label{eq:CBF_condition_rewritten}
\end{equation}

Given a CBF,~\cite{AmesXuGriTab2017} established safety-critical control.
\begin{theorem}[\cite{AmesXuGriTab2017, xu2015robustnessofCBFs}] \label{theo:CBF}
\textit{
If $h$ is a CBF for~\eqref{eq:system} on $\X$, then any locally Lipschitz controller $k$ that satisfies: 
\begin{equation}
    \dot{h} \big( x, k(x) \big) \geq - \alpha \big( h(x) \big)
\label{eq:safety_condition}
\end{equation}
for all ${x \in \C}$ renders~\eqref{eq:closedloop} safe w.r.t.~$\C$.
Furthermore, if~\eqref{eq:safety_condition} holds for all ${x \in \X}$, then $\C$ is asymptotically stable.
}
\end{theorem}
Accordingly, if the controller $k$ is synthesized such that~\eqref{eq:safety_condition} holds for all ${x \in \C}$, then the closed-loop system evolves in the safe set: ${x_0 \in \C \implies x(t) \in \C}$ ${\forall t \geq 0}$.
Moreover, even if the initial condition is outside $\C$, i.e., ${x_0 \notin \C}$, the system converges towards $\C$  if~\eqref{eq:safety_condition} is enforced for all ${x \in \X}$~\cite{xu2015robustnessofCBFs}.

Condition~\eqref{eq:safety_condition} is often used as constraint in optimization to synthesize safe controllers.
For example, a desired but not necessarily safe controller ${\kd: \X \to \U}$ can be modified to a safe controller via the {\em quadratic program (QP)}:
\begin{align}
\begin{split}
    k(x) = \underset{u \in \U}{\operatorname{argmin}} & \quad \| u - \kd(x) \|^2 \\[-3pt]
    \text{s.t.} & \quad \dot{h}(x,u) \geq - \alpha \big( h(x) \big),
\end{split}
\label{eq:QP}
\end{align}
also known as {\em safety filter}, which has explicit solution~\cite{Alan2023AV}:
\begin{align}
    k(x) & = \begin{cases}
        \kd(x) + \max\{0, \eta(x) \} \frac{\L{g}{h(x)}^\top}{\|\L{g}{h(x)}\|^2}, & {\rm if}\ \L{g}{h(x)} \neq 0, \\
        \kd(x), & {\rm if}\ \L{g}{h(x)} = 0,
    \end{cases} \nonumber \\
    \eta(x) & = -\L{f}{h(x)} - \L{g}{h(x)} \kd(x) - \alpha \big( h(x) \big).
\label{eq:QPsolu}
\end{align}

\subsection{Motivation: Multiple CBFs}

Controller~\eqref{eq:QPsolu} guarantees safety w.r.t.~a single safe set $\C$.
However, there exist more complex safety specifications in practice that involve compositions of multiple sets.
Such general specifications are discussed in the next section.
As motivation, we first consider the case of enforcing multiple safety constraints simultaneously, given by the sets:
\begin{equation}
    \C_i = \{ x \in \X: h_i(x) \geq 0 \},
\label{eq:safeset_individual}
\end{equation}
and CBF candidates $h_i$, with ${i \in I = \{1, 2, \ldots, N\}}$.
Our goal is to maintain ${x(t) \in \C_i}$ ${\forall t \geq 0}$ and ${\forall i \in I}$, that corresponds to rendering the {\em intersection} of sets $\C_i$ safe.

One may achieve this goal by enforcing multiple constraints on the input simultaneously,
for example, by the QP:
\begin{align}
\begin{split}
    k(x) = \underset{u \in \U}{\operatorname{argmin}} & \quad \| u - \kd(x) \|^2 \\[-3pt]
    \text{s.t.} & \quad \dot{h}_i(x,u) \geq - \alpha_i \big( h_i(x) \big) \quad \forall i \in I.
\end{split}
\label{eq:QPmultiple}
\end{align}
However,~\eqref{eq:QPmultiple} may not be feasible (its solution may not exist) for arbitrary number of constraints.
Even if each $h_i$ is CBF
and consequently each individual constraint in~\eqref{eq:QPmultiple} could be satisfied by a control input, the same input may not satisfy all constraints.
For the feasibility of~\eqref{eq:QPmultiple} we rather require:
\begin{equation}
    \sup_{u \in \U} \min_{i \in I} \Big( \dot{h}_i(x,u) + \alpha_i \big( h(x) \big) \Big) \geq 0,
\label{eq:CBF_condition_multiple_sup}
\end{equation}
cf.~\eqref{eq:CBF_condition}, that can also be stated in a form like~\eqref{eq:CBF_condition_rewritten} as follows.
\begin{theorem} \label{theo:feasibility}
\textit{
The QP~\eqref{eq:QPmultiple} is feasible if and only if:
\begin{equation}
    \!\sum_{i \in I}\! \lambda_i \L{g}{h_i(x)} \!=\! 0 \!\implies\! \!\!\sum_{i \in I}\! \lambda_i \Big(\! \L{f}{h_i(x)} \!+\! \alpha_i \big( h_i(x) \big) \!\Big) \!\geq\! 0\!
\label{eq:CBF_condition_multiple}
\end{equation}
holds for all ${x \in \X}$ and ${\lambda_i \geq 0}$.
}
\end{theorem}
\noindent The proof is given in the Appendix.

This highlights that multiple CBFs are more challenging to use than a single one.
With this as motivation, next we propose to encode all safety specifications into a single CBF.


\section{COMPLEX SAFETY SPECIFICATIONS}
\label{sec:multipleCBFs}

We propose a framework to construct a single CBF candidate that captures complex safety specifications, wherein safety is given by Boolean logical operations between multiple constraints.
For example, the motivation above involves logical AND operation: ${x(t) \in \C_1}$ AND ... AND ${x(t) \in \C_N}$ must hold.
Next, we discuss arbitrary logical compositions (with AND, OR and negation) of safety constraints.

\subsection{Operations Between Sets}

Consider multiple safety constraints, each given by a set $\C_i$
in~\eqref{eq:safeset_individual}.
These may be combined via the following Boolean logical operations to capture complex safety specifications.

\subsubsection{Identity / class-$\Ke$ function}
The 0-superlevel set $\C_i$ of $h_i$ is the same as that of $\gamma_i \circ h_i$ for any ${\gamma_i \in \Ke}$:
\begin{equation}
    \C_i = \{ x \in \X: \gamma_i \big( h_i(x) \big) \geq 0 \}.
\end{equation}

\subsubsection{Complement set / negation}
The complement\footnote{More precisely, $\overline{\C_i}$ is the closure of the complement of $\C_i$, i.e., it includes the boundary ${\partial \C_i}$ (where ${h_i(x) = 0}$).} $\overline{\C_i}$ of the 0-superlevel set of $h_i$ is the 0-superlevel set of $-h_i$:
\begin{equation}
    \overline{\C_i} = \{ x \in \X: -h_i(x) \geq 0 \}.
\end{equation}

\subsubsection{Union of sets / maximum / OR operation}

The union of multiple 0-superlevel sets:
\begin{equation}
    {\textstyle \bigcup_{i \in I}} \C_i = \{ x \in \X: \exists i \in I\ \text{s.t. } h_i(x) \geq 0 \}
\end{equation}
can be given by a single inequality with the $\max$ function~\cite{glotfelter2017nonsmooth}:
\begin{equation}
    {\textstyle \bigcup_{i \in I}} \C_i = \Big\{ x \in \X: \max_{i \in I} h_i(x) \geq 0 \Big\}.
\label{eq:maxcbf}
\end{equation}
The union describes logical OR relation between constraints: 
\begin{equation}
    x \!\in\! {\textstyle \bigcup_{i \in I}} \C_i \!\iff\! x \!\in\! \C_1\ \text{OR } x \!\in\! \C_2\ \ldots\ \text{OR } x \!\in\! \C_N.
\end{equation}

\subsubsection{Intersection of sets / minimum / AND operation}
The intersection of multiple 0-superlevel sets:
\begin{equation}
    {\textstyle \bigcap_{i \in I}} \C_i = \{ x \in \X: h_i(x) \geq 0 \;\; \forall i \in I \}
\end{equation}
can be compactly expressed using the $\min$ function~\cite{glotfelter2017nonsmooth}:
\begin{equation}
    {\textstyle \bigcap_{i \in I}} \C_i = \Big\{ x \in \X: \min_{i \in I} h_i(x) \geq 0 \Big\}.
\label{eq:mincbf}
\end{equation}
As in the motivation above, the intersection of sets captures logical AND relation between multiple safety constraints: 
\begin{equation}
    \!x \!\in\! {\textstyle \bigcap_{i \in I}} \C_i \!\iff\! x \!\in\! \C_1\ \text{AND } x \!\in\! \C_2\ \ldots\ \text{AND } x \!\in\! \C_N.\!
\end{equation}

Further operations between sets can be decomposed into applications of identity, complement, union and intersection, which are represented equivalently by class-$\Ke$ functions, negation, $\max$ and $\min$ operations, respectively.

\begin{remark} \label{rem:scaling}
Note that $h_i$ may have various physical meanings and orders of magnitude for different $i$.
Thus, for numerical conditioning (especially when we use exponentials later on), one may scale $h_i$ to ${\gamma_i \circ h_i}$ with continuously differentiable ${\gamma_i \in \Ke}$.
For example, ${\gamma_i(r) = \tanh(r)}$ scales to the interval ${\gamma_i(h_i(x)) \in [-1,1]}$ that may help numerics.
Next, we assume that the definitions of $h_i$ already include any necessary scaling and we omit $\gamma_i$.
Likewise, we do not discuss negation further by assuming that $h_i$ are defined with proper sign.
\end{remark}

\subsection{Smooth Approximations to Construct a Single CBF}

While the union and intersection of sets are described
by a single function in~\eqref{eq:maxcbf} and~\eqref{eq:mincbf}, the resulting expressions, ${\max_{i \in I} h_i(x)}$ and ${\min_{i \in I} h_i(x)}$, may not be continuously differentiable in $x$ \cite{glotfelter2017nonsmooth}, and they are not CBFs.
As main result, we propose a CBF candidate by smooth approximations of $\max$ and $\min$, and describe its properties.
This enables us to enforce complex safety specifications as a single constraint.

\subsubsection{Union of Sets}


To capture the union of sets in~\eqref{eq:maxcbf}, we propose a CBF candidate via a smooth over-approximation of the $\max$ function using a log-sum-exp expression~\cite{lindemann2019stl}:
\begin{equation}
    h(x) = \frac{1}{\kappa} \ln \bigg( \sum_{i \in I} {\rm e}^{\kappa h_i(x)} \bigg)
\label{eq:hmax}
\end{equation}
with smoothing parameter ${\kappa>0}$.
The Lie derivatives are:
\begin{equation}
    \!\L{f}{h(x)} \!\!=\!\! \!\sum_{i \in I}\! \lambda_i(x) \!\L{f}{h_i(x)}, \;
    \L{g}{h(x)} \!\!=\!\! \!\sum_{i \in I}\! \lambda_i(x) \!\L{g}{h_i(x)},\!
\label{eq:Lfgh_multiple}
\end{equation}
with the coefficients:
\begin{equation}
    \lambda_i(x) = {\rm e}^{\kappa (h_i(x) - h(x))},
\label{eq:maxcoeff}
\end{equation}
that satisfy ${\sum_{i \in I}\lambda_i(x) = 1}$.
The proposed CBF candidate in~\eqref{eq:hmax} has the properties below; see proof in the Appendix.

\begin{theorem} \label{theo:maxCBF}
\textit{
Consider sets $\C_i$ in~\eqref{eq:safeset_individual} given by functions $h_i$, and the union ${\bigcup_{i \in I} \C_i}$ in~\eqref{eq:maxcbf}.
Function $h$ in~\eqref{eq:hmax} over-approximates the $\max$ expression in~\eqref{eq:maxcbf} with bounds:
\begin{equation}
    \max_{i \in I} h_i(x) \leq h(x) \leq \max_{i \in I} h_i(x) + \frac{\ln N}{\kappa} \quad \forall x \in \mathbb{R}^n,
\label{eq:hmax_bounds}
\end{equation}
such that ${\lim_{\kappa \to \infty} h(x) = \max_{i \in I} h_i(x)}$.
The corresponding set $\C$ in~\eqref{eq:safeset} encapsulates the union, ${\C \supseteq \bigcup_{i \in I} \C_i}$,
such that ${\lim_{\kappa \to \infty} \C = \bigcup_{i \in I} \C_i}$.
Moreover, if~\eqref{eq:CBF_condition_multiple} holds for all ${x \in \X}$
with
$\lambda_i$ in~\eqref{eq:maxcoeff}, then $h$ is a CBF for~\eqref{eq:system} on $\X$ with any ${\alpha \in \Keinf}$ that satisfies
${\alpha(r) \geq \alpha_i(r)}$ 
${\forall r \in \R}$ and ${\forall i \in I}$.
}
\end{theorem}


\begin{remark} \label{rem:buffer}
A set $\C$ that {\em lies inside} the union of the individual sets can also be built by using a buffer $b$ when defining $h$:
\begin{equation}
    h(x) = \frac{1}{\kappa} \ln \bigg( \sum_{i \in I} {\rm e}^{\kappa h_i(x)} \bigg) - \frac{b}{\kappa}.
\label{eq:hmax_buffer}
\end{equation}
For example, based on the upper bound in~\eqref{eq:hmax_bounds}, ${b = \ln N}$ leads to ${h(x) \leq \max_{i \in I} h_i(x)}$ and ${\C \subseteq \bigcup_{i \in I} \C_i}$.
Alternatively, buffers from problem-specific bounds that are tighter than~\eqref{eq:hmax_bounds} can give better inner-approximation $\C$ of ${\bigcup_{i \in I} \C_i}$.
\end{remark}


\begin{figure}
\centering
\includegraphics[scale=1]{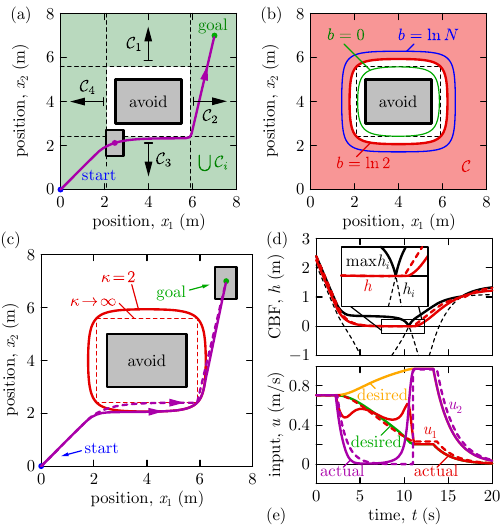}
\vspace{-3mm}
\caption{
Numerical results for Example~\ref{ex:obstacle_single}, where a reach-avoid task is safely executed.
(a) Safe set,
(b) 0-superlevel set of the proposed CBF~\eqref{eq:hmax_buffer},
(c)-(e) simulation of safety-critical control by~\eqref{eq:QPsolu}.
}
\vspace{-5mm}
\label{fig:obstacle_single}
\end{figure}

\begin{example} \label{ex:obstacle_single}
Consider Fig.~\ref{fig:obstacle_single}, where a rectangular agent with planar position ${x \in \R^2}$, velocity ${u \in \R^2}$, and dynamics:
\begin{equation}
    \dot{x} = u
\label{eq:single_integrator}
\end{equation}
is controlled to reach a desired position ${\xd \in \R^2}$ while avoiding a rectangular obstacle\footnote{Matlab codes for each example are available at: \url{https://github.com/molnartamasg/CBFs-for-complex-safety-specs}.}.
To reach the goal, we use a proportional controller with gain ${\Kp>0}$ and saturation: 
\begin{equation}
    \kd(x) = \sat \big( \Kp (\xd-x) \big),
\label{eq:proportional_controller}
\end{equation}
where ${\sat(u) = \min \{ 1, \umax / \|u\|_2 \} u}$ with some ${\umax>0}$.
We modify this desired controller to a safe controller using the safety filter~\eqref{eq:QPsolu} and the proposed CBF construction.

To avoid the obstacle, the agent's center must be outside a rectangle that has the combined size of the obstacle and the agent; see Fig.~\ref{fig:obstacle_single}(a).
This means ${N=4}$ constraints linked with OR logic: keep the center left to OR above OR right to OR below the rectangle.
Accordingly, the safe set is given by the union ${\bigcup_{i \in I} \C_i}$ of four individual sets $\C_i$ described by four barriers at location ${x_i \in \R^2}$ with normal vector ${n_i \in \R^2}$:
\begin{equation}
    h_i(x) = n_i^\top (x-x_i),
\label{eq:barrier_halfspace}
\end{equation}
${i \in I = \{1, 2, 3, 4\}}$.
We combine the four barriers with~\eqref{eq:hmax_buffer}.
The resulting safe set $\C$ is plotted in Fig.~\ref{fig:obstacle_single}(b) for ${\kappa = 2}$ and various buffers $b$.
Set $\C$ encapsulates ${\bigcup_{i \in I} \C_i}$ for ${b = 0}$, whereas
set $\C$ lies inside ${\bigcup_{i \in I} \C_i}$ for ${b = \ln N}$; cf.~Remark~\ref{rem:buffer}.
For the problem-specific buffer ${b = \ln 2}$ (where $N$ is replaced by $2$ since two barriers meet at each corner), the approximation $\C$ gets very close to the corners of ${\bigcup_{i \in I} \C_i}$.

We executed controller~\eqref{eq:QPsolu} with ${\Kp=0.5}$, ${\umax = 1}$, ${\kappa=2}$, ${b = \ln 2}$ and ${\alpha(h) = h}$; see solid lines in Fig.~\ref{fig:obstacle_single}(c).
The reach-avoid task is successfully accomplished by keeping the agent within set $\C$.
Fig.~\ref{fig:obstacle_single}(d) highlights that safety is maintained w.r.t.~a smooth under-approximation $h$ (red) of the maximum ${\max_{i \in I} h_i}$ (black) of the individual barriers $h_i$ (dashed).
Fig.~\ref{fig:obstacle_single}(e) indicates the underlying control input.
We also demonstrate by dashed lines in Fig.~\ref{fig:obstacle_single}(c)-(e) the case of increasing the smoothing parameter to ${\kappa \to \infty}$.
The sharp corner is recovered and the input becomes discontinuous ($u_2$ jumps).
While discontinuous inputs can be addressed by nontrivial nonsmooth CBF theory~\cite{glotfelter2017nonsmooth}, they may be difficult to realize accurately by actuators in engineering systems.
\end{example}

\subsubsection{Intersection of Sets}

To capture the intersection of sets in~\eqref{eq:mincbf}, we propose to use a smooth under-approximation of the $\min$ function as CBF candidate~\cite{lindemann2019stl}, analogously to~\eqref{eq:hmax}:
\begin{equation}
    h(x) = - \frac{1}{\kappa} \ln \bigg( \sum_{i \in I} {\rm e}^{- \kappa h_i(x)} \bigg).
\label{eq:hmin}
\end{equation}
The Lie derivatives of $h$ are expressed by~\eqref{eq:Lfgh_multiple} with:
\begin{equation}
    \lambda_i(x) = {\rm e}^{- \kappa (h_i(x) - h(x))},
\label{eq:mincoeff}
\end{equation}
that satisfy ${\sum_{i \in I}\lambda_i(x) = 1}$.
The proposed CBF candidate in~\eqref{eq:hmin} has the properties below, as proven in the Appendix.

\begin{theorem} \label{theo:minCBF}
\textit{
Consider sets $\C_i$ in~\eqref{eq:safeset_individual} given by functions $h_i$, and the intersection ${\bigcap_{i \in I} \C_i}$ in~\eqref{eq:mincbf}.
Function $h$ in~\eqref{eq:hmin} under-approximates the $\min$ expression in~\eqref{eq:mincbf} with bounds:
\begin{equation}
    \min_{i \in I} h_i(x) - \frac{\ln N}{\kappa} \leq h(x) \leq \min_{i \in I} h_i(x) \quad \forall x \in \mathbb{R}^n,
\label{eq:hmin_bounds}
\end{equation}
such that ${\lim_{\kappa \to \infty} h(x) = \min_{i \in I} h_i(x)}$.
The corresponding set $\C$ in~\eqref{eq:safeset} lies inside the intersection, ${\C \subseteq \bigcap_{i \in I} \C_i}$,
such that ${\lim_{\kappa \to \infty} \C = \bigcap_{i \in I} \C_i}$.
}
\end{theorem}


\subsection{Single CBF for Arbitrary Safe Set Compositions}

Having discussed the union and intersection of sets, we extend our framework to arbitrary combinations of unions and intersections.
These include e.g. two-level or three-level compositions, like ${\bigcup \bigcap_{i} \C_i}$ or ${\bigcap \bigcup \bigcap_{i} \C_i}$, etc.
We propose an algorithmic way to capture these by a single CBF candidate.

Specifically, consider $M$ levels of safety specifications that establish a single safe set by composing $N$ individual sets.
The individual sets are $\C_i$ in~\eqref{eq:safeset_individual}, ${i \in I = \{1, \ldots, N\}}$.
The specification levels are indexed by ${\ell \in L = \{1, \ldots, M\}}$.
At each level, the union or intersection of sets is taken, resulting in $N_\ell$ new sets, denoted by $\C_i^\ell$, ${i \in I_\ell = \{1, \ldots, N_\ell\}}$.
This is repeated until a  single safe set, called $\Cc$, is obtained:
\begin{align}
\begin{split}
    \C_i^0 & = \C_i,
    \quad i \in I, \\
    \C_i^\ell & = \begin{cases}
    \bigcup_{j \in J_i^\ell} \C_j^{\ell-1} & {\rm if}\ \ell \in \Lor, \\
    \bigcap_{j \in J_i^\ell} \C_j^{\ell-1} & {\rm if}\ \ell \in \Land, 
    \end{cases}
    \quad i \in I_\ell, \\
    \Cc & = \C_1^M,
\end{split}
\label{eq:combination_set}
\end{align}
where ${J_i^\ell \subseteq I_{\ell-1}}$ is the indices of sets that combine into $C_i^\ell$, while $\Lor$ and $\Land$ are the indices of levels with union and intersection (${L = \Lor \cup \Land}$).
Unions and intersections imply the maximum and minimum of the individual barriers $h_i$, respectively, resulting in the combined CBF candidate $\hc$~\cite{glotfelter2017nonsmooth}:
\begin{align}
\begin{split}
    h_i^0(x) & = h_i(x),
    \quad i \in I, \\
    h_i^\ell(x) & = \begin{cases}
    \max_{j \in J_i^\ell} h_j^{\ell-1}(x) & {\rm if}\ \ell \in \Lor, \\
    \min_{j \in J_i^\ell} h_j^{\ell-1}(x) & {\rm if}\ \ell \in \Land, 
    \end{cases}
    \quad i \in I_\ell, \\
    \hc(x) & = h_1^M(x).
\end{split}
\label{eq:combination_CBF_nonsmooth}
\end{align}
This describes the safe set (that is assumed to be non-empty):
\begin{equation}
    \Cc = \{ x \in \X: \hc(x) \geq 0 \}.
\label{eq:combination_set_superlevel}
\end{equation}

While the combined function $\hc$ is nonsmooth~\cite{glotfelter2017nonsmooth}, we propose a continuously differentiable function $h$, by extending the smooth approximations~\eqref{eq:hmax} and~\eqref{eq:hmin} of min and max:
\begin{align}
\begin{split}
    H_i^0(x) & = {\rm e}^{\kappa h_i(x)},
    \quad i \in I, \\
    H_i^\ell(x) & = \begin{cases}
    \sum_{j \in J_i^\ell} H_j^{\ell-1}(x) & {\rm if}\ \ell \in \Lor, \\
    \frac{1}{\sum_{j \in J_i^\ell} \frac{1}{H_j^{\ell-1}(x)}} & {\rm if}\ \ell \in \Land,
    \end{cases}
    \quad i \in I_\ell, \\
    h(x) & = \frac{1}{\kappa} \ln H_1^M(x) - \frac{b}{\kappa}.
\end{split}
\label{eq:combination_CBF}
\end{align}
Note that we included a buffer $b$, according to Remark~\ref{rem:buffer}, to be able to adjust whether the resulting set $\C$ encapsulates $\Cc$ or lies inside it.
The derivative of the CBF candidate $h$
is:
\begin{align}
    \dot{H}_i^0(x,u) & = \kappa H_i^0(x) \dot{h}_i(x,u),
    \quad i \in I, \nonumber \\
    \dot{H}_i^\ell(x,u) & = \begin{cases}
    \sum_{j \in J_i^\ell} \dot{H}_j^{\ell-1}(x,u) \!&\! {\rm if}\ \ell \in \Lor, \\
    H_i^\ell(x)^2 \sum_{j \in J_i^\ell} \frac{\dot{H}_j^{\ell-1}(x,u)}{H_j^{\ell-1}(x)^2} \!&\! {\rm if}\ \ell \in \Land,
    \end{cases}
    \; i \in I_\ell, \nonumber \\
    \dot{h}(x,u) & = \frac{\dot{H}_1^M(x,u)}{\kappa H_1^M(x)}.
\label{eq:combination_CBF_derivative}
\end{align}

The proposed function $h$ approximates $\hc$ with the following properties that are proven in the Appendix.
\begin{theorem} \label{theo:error}
\textit{
Consider sets $\C_i$ in~\eqref{eq:safeset_individual} given by functions $h_i$, and the composition $\Cc$ in~\eqref{eq:combination_set} given by $\hc$ in~\eqref{eq:combination_CBF_nonsmooth}-\eqref{eq:combination_set_superlevel}.
Function $h$ in~\eqref{eq:combination_CBF} approximates $\hc$ with the error bound:
\begin{equation}
    - \frac{b_{\cap} + b}{\kappa} \leq h(x) - \hc(x) \leq \frac{b_{\cup} - b}{\kappa}  \quad \forall x \in \mathbb{R}^n,
\label{eq:combination_error}
\end{equation}
where
${b_{\cap} \!=\! \sum_{\ell \in \Land}\! \ln b_\ell}$,
${b_{\cup} \!=\! \sum_{\ell \in \Lor}\! \ln b_\ell}$,
${b_\ell \!=\! \max_{i \in I_\ell}\! |J_i^\ell|}$,
and $|J_i^\ell|$ is the number of elements in $J_i^\ell$.
If ${b \geq b_{\cup}}$, the corresponding set $\C$ in~\eqref{eq:safeset} lies inside $\Cc$, i.e., ${\C \subseteq \Cc}$, whereas if ${b \leq - b_{\cap}}$, set $\C$ encapsulates $\Cc$, i.e., ${\C \supseteq \Cc}$.
Furthermore, we have ${\lim_{\kappa \to \infty} h(x) = \hc(x)}$ and ${\lim_{\kappa \to \infty} \C = \Cc}$.
}
\end{theorem}

The proposed approach in~\eqref{eq:combination_CBF} captures complex safety specifications algorithmically by a single CBF candidate $h$, via the recursive use of~\eqref{eq:hmax} and~\eqref{eq:hmin} such that exponentials and logarithms are computed only once.
Safety is then interpreted w.r.t.~set $\C$, which can be tuned to approximate the specified set $\Cc$ as desired.
Based on the error bound~\eqref{eq:combination_error}, increasing $\kappa$ makes the approximation tighter, while $b$ affects whether ${\C \subseteq \Cc}$ or ${\C \supseteq \Cc}$.
Note that $h$ is a valid CBF only if it satisfies~\eqref{eq:CBF_condition}.
This is not guaranteed by Theorem~\ref{theo:error}, and it would require additional conditions like~\eqref{eq:CBF_condition_multiple} in Theorem~\ref{theo:maxCBF}.
If $h$ is a CBF, formal safety guarantees can be maintained, for example, by QP~\eqref{eq:QP} that has a single constraint and the explicit solution~\eqref{eq:QPsolu}.
If the constraint is enforced outside set $\C$, then $\C$ is asymptotically stable; cf.~Theorem~\ref{theo:CBF}.
We remark that, potentially, the log-sum-exp formulas could be replaced by other smooth approximations of $\max$ and $\min$.
Furthermore, note that computing exponentials may cause numerical issues if $\kappa$ is too large.
These may be alleviated by scaling CBF candidates by class-$\Ke$ functions; see Remark~\ref{rem:scaling}.






\begin{figure}
\centering
\includegraphics[scale=1]{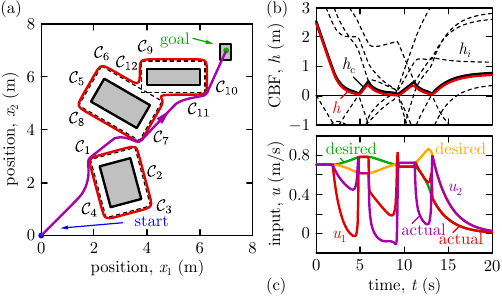}
\vspace{-3mm}
\caption{
Numerical results for Example~\ref{ex:obstacle_multiple}, where a reach-avoid task with multiple obstacles is executed by controller~\eqref{eq:QPsolu} with the proposed CBF~\eqref{eq:combination_CBF}.
}
\vspace{-6mm}
\label{fig:obstacle_multiple}
\end{figure}

\begin{example} \label{ex:obstacle_multiple}
Consider the reach-avoid task of Example~\ref{ex:obstacle_single}, with dynamics~\eqref{eq:single_integrator}, desired controller~\eqref{eq:proportional_controller}, safety filter~\eqref{eq:QPsolu}, and multiple obstacles shown in Fig.~\ref{fig:obstacle_multiple}.
Like in Example~\ref{ex:obstacle_single}, each of the three obstacles yields four safety constraints, leading to ${N=12}$ sets $\C_i$ and functions $h_i$, given by~\eqref{eq:barrier_halfspace}.
The four constraints of each obstacle are linked with OR logic, like in Example~\ref{ex:obstacle_single}, while the constraints of different obstacles are linked with AND: safety is maintained w.r.t.~obstacle~1 AND obstacle~2 AND obstacle~3.
Thus, the safe set:
\begin{equation}
    \Cc \!=\! (\C_1 \cup \C_2 \cup \C_3 \cup \C_4)
    \cap (\C_5 \cup \C_6 \cup \C_7 \cup \C_8)
    \cap (\C_9 \cup \C_{10} \cup \C_{11} \cup \C_{12})
\label{eq:combination_example}
\end{equation}
is given by a ${M=2}$ level specification, combining ${N = 12}$ sets to ${N_1 = 3}$ sets
($\C_1^1$ from sets given by ${J_1^1 = \{1,2,3,4\}}$,
$\C_2^1$ from ${J_2^1 = \{5,6,7,8\}}$ and
$\C_3^1$ from ${J_3^1 = \{9,10,11,12\}}$),
and then to a single set $\Cc$ (via sets given by ${J_1^2 = \{1,2,3\}}$).

The behavior of controller~\eqref{eq:QPsolu} with the proposed CBF candidate~\eqref{eq:combination_CBF} is shown in Fig.~\ref{fig:obstacle_multiple} for ${\Kp=0.5}$, ${\umax = 1}$, ${\kappa=10}$, ${b = \ln 2}$ and ${\alpha(h) = h}$.
The reach-avoid task is successfully accomplished with formal guarantees of safety.
Remarkably, the controller is continuous and explicit, since the control law~\eqref{eq:QPsolu} and CBF formulas~\eqref{eq:combination_CBF}-\eqref{eq:combination_CBF_derivative} are in closed form.
Such explicit controllers are easy to implement and fast to execute.
Note that controller~\eqref{eq:QPmultiple} could also handle multiple obstacles if each obstacle was given by a single CBF candidate.
Yet,~\eqref{eq:QPmultiple} cannot address multi-level safety specifications like~\eqref{eq:combination_example}, while the proposed method can.
\end{example}

\begin{figure}
\centering
\includegraphics[scale=1]{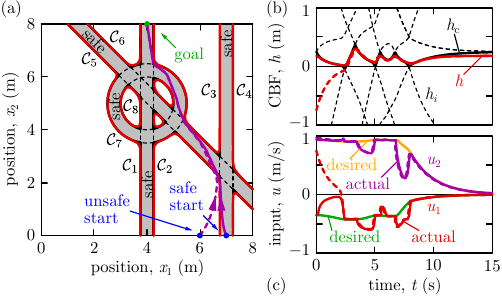}
\vspace{-3mm}
\caption{
Numerical results for Example~\ref{ex:road}, where an agent is driven safely along a road network via controller~\eqref{eq:QPsolu} with the proposed CBF~\eqref{eq:combination_CBF}.
}
\vspace{-6mm}
\label{fig:road}
\end{figure}

\begin{example} \label{ex:road}
Consider the setup of Fig.~\ref{fig:road} where a point agent is driven to a desired location while staying on a road network, with dynamics~\eqref{eq:single_integrator}, desired controller~\eqref{eq:proportional_controller} and safety-critical controller~\eqref{eq:QPsolu}.
Safety is determined by the road geometry.
Each road boundary is related to a set, which is given for straight roads by~\eqref{eq:barrier_halfspace} and for ring roads by:
\begin{equation}
    h_i(x) = \pm \big( \| x - x_i \| - R_i \big).
\end{equation}
Here plus and minus signs stand for the inner and outer circles, respectively, $R_i$ is their radius, and $x_i$ is their center.
Safety must be ensured w.r.t.~boundary~1 AND boundary~2 of each road, while the agent must stay on road~1 OR road~2 OR road~3 OR road~4.
Thus, the combined safe set becomes:
\begin{equation}
    \Cc =
    (\C_1 \cap \C_2) \cup
    (\C_3 \cap \C_4) \cup
    (\C_6 \cap \C_5) \cup
    (\C_7 \cap \C_8).
\end{equation}
That is, we have a ${M=2}$ level specification with ${N = 8}$ sets combined first to ${N_1 = 4}$ sets
(as intersections of sets given by ${J_1^1 = \{1,2\}}$, ${J_2^1 = \{3,4\}}$, ${J_3^1 = \{5,6\}}$, ${J_4^1 = \{7,8\}}$), and then to a single set (as union via ${J_1^2 = \{1,2,3,4\}}$).

The execution of the reach-avoid task with the proposed CBF candidate~\eqref{eq:combination_CBF} and controller~\eqref{eq:QPsolu} is shown in Fig.~\ref{fig:road} for ${\Kp=0.5}$, ${\umax = 1}$, ${\kappa=10}$, ${b = 0}$ and ${\alpha(h) = h}$.
The end result is guaranteed safety (see solid lines).
Moreover, the safe set is attractive: in case of an unsafe, off-road initial condition the agent returns to to the safe set on the road and continues to be safe (see thick dashed lines).
Remarkably, this property was not provided by earlier works like~\cite{wang2016multiobjective}.
\end{example}


\section{CONCLUSION}
\label{sec:concl}

We established a framework to capture complex safety specifications by control barrier functions (CBFs).
The specifications are combinations of state constraints by Boolean logic.
We proposed an algorithmic way to create a single CBF candidate that encodes these constraints and enables efficient safety-critical controllers.
We described the properties of this CBF candidate, and we used simulations to show its ability to tackle nontrivial safety-critical control problems.

\section*{APPENDIX}

\begin{proof}[Proof of Theorem~\ref{theo:feasibility}]
%
Consider the Lagrangian of the feasibility problem~\cite{boyd2004convex} corresponding to the QP~\eqref{eq:QPmultiple}:
\begin{equation}
    L(x,u,\lambda) = - \sum_{i \in I} \lambda_i \Big( \dot{h}_i(x,u) + \alpha_i \big( h(x) \big) \Big),
\end{equation}
with the Lagrange multipliers ${\lambda \!=\! \begin{bmatrix} \lambda_1 \!\!&\!\! \lambda_2 \!\!&\!\! \ldots \!\!&\!\! \lambda_N \end{bmatrix}^\top}$, ${\lambda_i \geq 0}$ ${\forall i \in I}$.
The QP~\eqref{eq:QPmultiple} is feasible if and only if ${\exists u \in \U}$ such that ${L(x,u,\lambda) \leq 0}$ ${\forall \lambda_i \geq 0}$.
With the Lagrange dual function, ${g_{L}(x,\lambda) \!=\! \inf_{u \in \U} L(x,u,\lambda)}$, this means ${g_{L}(x,\lambda) \leq 0}$ ${\forall \lambda_i \geq 0}$.
Since ${g_{L}(x,\lambda) \!=\! -\sum_{i \in I}\! \lambda_i \Big(\! \L{f}{h_i(x)} \!+\! \alpha_i \big( h_i(x) \big) \!\Big)}$ if ${\sum_{i \in I}\! \lambda_i \L{g}{h_i(x)} \!=\! 0}$
and ${g_{L}(x,\lambda) \!=\! -\infty}$ otherwise,~\eqref{eq:CBF_condition_multiple} is equivalent to~${g_{L}(x,\lambda) \leq 0}$ and provides feasibility.
\end{proof}

\begin{proof}[Proof of Theorem~\ref{theo:maxCBF}]
Since the exponential function is monotonous and gives positive value, we have:
\begin{equation}
    {\rm e}^{\kappa \max_{i \in I} h_i(x)} \leq \sum_{i \in I} {\rm e}^{\kappa h_i(x)} \leq
    N {\rm e}^{\kappa \max_{i \in I} h_i(x)},
\label{eq:exp_bound_max}
\end{equation}
that yields~\eqref{eq:hmax_bounds} via~\eqref{eq:hmax} and the monotonicity of $\ln$.
The limit on both sides of~\eqref{eq:hmax_bounds} yields ${\lim_{\kappa \to \infty} h(x) = \max_{i \in I} h_i(x)}$, and consequently ${\lim_{\kappa \to \infty} \C = \bigcup_{i \in I} \C_i}$ holds.
Due to~\eqref{eq:hmax_bounds},
${\max_{i \in I} h_i(x) \geq 0 \implies h(x) \geq 0}$,
therefore
${x \in \bigcup_{i \in I} \C_i \implies x \in \C}$, and ${\C \supseteq \bigcup_{i \in I} \C_i}$ follows.

We prove that $h$ is a CBF by showing that~\eqref{eq:CBF_condition_rewritten} holds.
We achieve this by relating $\L{g}{h(x)}$ and ${\L{f}{h(x)} + \alpha \big( h(x) \big)}$ to $\L{g}{h_i(x)}$ and ${\L{f}{h_i(x)} + \alpha_i \big( h_i(x) \big)}$.
The Lie derivatives are related by~\eqref{eq:Lfgh_multiple},
while the following bound holds for all ${i \in I}$:
\begin{equation}
    \alpha \big( h(x) \big)
    \geq \alpha \big( h_i(x) \big)
    \geq \alpha_i \big( h_i(x) \big),
\label{eq:alpha_condition_rewritten}
\end{equation}
where we used~\eqref{eq:hmax_bounds}
and ${\alpha(r) \geq \alpha_i(r)}$.
Consequently, since ${\sum_{i \in I} \lambda_i(x) = 1}$ and ${\lambda_i(x)>0}$ hold via~\eqref{eq:maxcoeff}, we have:
\begin{equation}
    \L{f}{h(x)} \!+\! \alpha \big( h(x) \big) \geq \sum_{i \in I} \lambda_i(x) \Big( \L{f}{h_i(x)} \!+\! \alpha_i \big( h_i(x) \big) \Big).
\label{eq:Lfh_relationship_max}
\end{equation}
If ${\L{g}{h(x)} = 0}$, we get ${\sum_{i \in I} \lambda_i(x) \L{g}{h_i(x)} = 0}$ based on~\eqref{eq:Lfgh_multiple}, and since~\eqref{eq:CBF_condition_multiple} is assumed to hold,~\eqref{eq:Lfh_relationship_max} finally yields ${\L{f}{h(x)} + \alpha \big( h(x) \big) \geq 0}$.
Thus,~\eqref{eq:CBF_condition_rewritten} holds and $h$ is a CBF.
\end{proof}

\begin{proof}[Proof of Theorem~\ref{theo:minCBF}]
The proof follows that of Theorem~\ref{theo:maxCBF}, with the following modifications.
We replace~\eqref{eq:exp_bound_max} by:
\begin{equation}
    {\rm e}^{-\kappa \min_{i \in I} h_i(x)} \leq \sum_{i \in I} {\rm e}^{-\kappa h_i(x)} \leq
    N {\rm e}^{-\kappa \min_{i \in I} h_i(x)},
\label{eq:exp_bound_min}
\end{equation}
that gives the bound~\eqref{eq:hmin_bounds} via~\eqref{eq:hmin}.
The remaining properties follow from the limit on both sides of~\eqref{eq:hmin_bounds} and from
${h(x) \geq 0 \implies \min_{i \in I} h_i(x) \geq 0}$ according to~\eqref{eq:hmin_bounds}.
\end{proof}

\begin{proof}[Proof of Theorem~\ref{theo:error}]
By leveraging that the exponential function is monotonous,
we write~\eqref{eq:combination_CBF_nonsmooth} equivalently as:
\begin{align}
\begin{split}
    H_{{\rm c},i}^0(x) & = {\rm e}^{\kappa h_i(x)},
    \quad i \in I, \\
    H_{{\rm c},i}^\ell(x) & = \begin{cases}
    \max_{j \in J_i^\ell} H_{{\rm c},j}^{\ell-1}(x) & {\rm if}\ \ell \in \Lor, \\
    \min_{j \in J_i^\ell} H_{{\rm c},j}^{\ell-1}(x) & {\rm if}\ \ell \in \Land, 
    \end{cases}
    \;\; i \in I_\ell, \\
    \hc(x) & = \frac{1}{\kappa} \ln H_{{\rm c},1}^M(x).
\end{split}
\label{eq:combination_CBF_nonsmooth_log}
\end{align}
We compare this with the definition~\eqref{eq:combination_CBF} of $h$.
First, by using the middle row of~\eqref{eq:combination_CBF}, we establish that for all ${x \in \X}$:
\begin{align}
\begin{split}
    H_j^{\ell-1}(x) \leq
    H_i^\ell(x) \leq
    |J_i^\ell| \max_{j \in J_i^\ell} H_j^{\ell-1}(x)
    \quad {\rm if}\ \ell \in \Lor, \\
    \frac{1}{|J_i^\ell|} \min_{j \in J_i^\ell} H_j^{\ell-1}(x) \leq
    H_i^\ell(x) \leq
    H_j^{\ell-1}(x)
    \quad {\rm if}\ \ell \in \Land
\end{split}
\label{eq:exp_bound}
\end{align}
${\forall j \!\in\! J_i^\ell}$ and ${\forall i \!\in\! I_\ell}$.
Then, we relate $H_{{\rm c},i}^\ell$ to $H_i^\ell$ by induction.
For ${\ell \!\geq\! 1}$ we assume that there exist ${\underline{c}_{\ell-1}, \overline{c}_{\ell-1} > 0}$ such that:
\begin{equation}
    \underline{c}_{\ell-1} H_{{\rm c},i}^{\ell-1}(x) \leq
    H_i^{\ell-1}(x) \leq
    \overline{c}_{\ell-1} H_{{\rm c},i}^{\ell-1}(x)
\label{eq:induction_assumption}
\end{equation}
${\forall x \in \X}$ and
${\forall i \in I_{\ell-1}}$.
This is true for ${\ell\!=\!1}$ with ${\underline{c}_0, \overline{c}_0 = 1}$ since ${H_i^0(x) = H_{{\rm c},i}^0(x)}$.
By substituting~\eqref{eq:induction_assumption} into~\eqref{eq:exp_bound}, using the middle row of~\eqref{eq:combination_CBF_nonsmooth_log} and ${|J_i^\ell| \leq \max_{i \in I_\ell} |J_i^\ell|}$,
we get:
\begin{equation}
    \underline{c}_\ell H_{{\rm c},i}^\ell(x) \leq
    H_i^\ell(x) \leq
    \overline{c}_\ell H_{{\rm c},i}^\ell(x)
\label{eq:induction_bound}
\end{equation}
with ${b_\ell = \max_{i \in I_\ell} |J_i^\ell|}$ and:
\begin{align}
\begin{split}
    \underline{c}_\ell =
    \begin{cases}
        \underline{c}_{\ell-1} \!& {\rm if}\ \ell \in \Lor, \\
        \frac{\underline{c}_{\ell-1}}{b_\ell} \!& {\rm if}\ \ell \in \Land,
    \end{cases} \quad
    \overline{c}_\ell =
    \begin{cases}
        b_\ell \overline{c}_{\ell-1} \!& {\rm if}\ \ell \in \Lor, \\
        \overline{c}_{\ell-1} \!& {\rm if}\ \ell \in \Land.
    \end{cases}
\end{split}
\label{eq:induction_coefficient}
\end{align}
By induction,~\eqref{eq:induction_bound} holds for ${\ell=M}$ with
${\underline{c}_M \!=\! \prod_{\ell \in \Land}\! \frac{1}{b_\ell}}$ and
${\overline{c}_M \!=\! \prod_{\ell \in \Lor}\! b_\ell}$.
Taking the logarithm of~\eqref{eq:induction_bound} with ${\ell=M}$ and using the last rows of~\eqref{eq:combination_CBF} and~\eqref{eq:combination_CBF_nonsmooth_log} result in~\eqref{eq:combination_error}.
\end{proof}

% ...


\bibliographystyle{IEEEtran}
\bibliography{refs}	

\begin{thebibliography}{10}
\providecommand{\url}[1]{#1}
\csname url@rmstyle\endcsname
\providecommand{\newblock}{\relax}
\providecommand{\bibinfo}[2]{#2}
\providecommand\BIBentrySTDinterwordspacing{\spaceskip=0pt\relax}
\providecommand\BIBentryALTinterwordstretchfactor{4}
\providecommand\BIBentryALTinterwordspacing{\spaceskip=\fontdimen2\font plus
\BIBentryALTinterwordstretchfactor\fontdimen3\font minus
  \fontdimen4\font\relax}
\providecommand\BIBforeignlanguage[2]{{%
\expandafter\ifx\csname l@#1\endcsname\relax
\typeout{** WARNING: IEEEtran.bst: No hyphenation pattern has been}%
\typeout{** loaded for the language `#1'. Using the pattern for}%
\typeout{** the default language instead.}%
\else
\language=\csname l@#1\endcsname
\fi
#2}}

\bibitem{AmesXuGriTab2017}
A.~D. Ames, X.~Xu, J.~W. Grizzle, and P.~Tabuada, ``Control barrier function
  based quadratic programs for safety critical systems,'' \emph{IEEE
  Transactions on Automatic Control}, vol.~62, no.~8, pp. 3861--3876, 2017.

\bibitem{rauscher2016constrained}
M.~Rauscher, M.~Kimmel, and S.~Hirche, ``Constrained robot control using
  control barrier functions,'' in \emph{IEEE/RSJ International Conference on
  Intelligent Robots and Systems}, 2016, pp. 279--285.

\bibitem{xu2018constrained}
X.~Xu, ``Constrained control of input--output linearizable systems using
  control sharing barrier functions,'' \emph{Automatica}, vol.~87, pp.
  195--201, 2018.

\bibitem{shawcortez2022robust}
W.~Shaw~Cortez, X.~Tan, and D.~V. Dimarogonas, ``A robust, multiple control
  barrier function framework for input constrained systems,'' \emph{IEEE
  Control Systems Letters}, vol.~6, pp. 1742--1747, 2022.

\bibitem{tan2022compatibility}
X.~Tan and D.~V. Dimarogonas, ``Compatibility checking of multiple control
  barrier functions for input constrained systems,'' in \emph{61st IEEE
  Conference on Decision and Control}, 2022, pp. 939--944.

\bibitem{breeden2023compositions}
J.~Breeden and D.~Panagou, ``Compositions of multiple control barrier functions
  under input constraints,'' in \emph{American Control Conference}, 2023, pp.
  3688--3695.

\bibitem{liu2020blf}
L.~Liu, Y.-J. Liu, D.~Li, S.~Tong, and Z.~Wang, ``Barrier {Lyapunov}
  function-based adaptive fuzzy {FTC} for switched systems and its applications
  to resistance–inductance–capacitance circuit system,'' \emph{IEEE
  Transactions on Cybernetics}, vol.~50, no.~8, pp. 3491--3502, 2020.

\bibitem{liu2021adaptive}
L.~Liu, W.~Zhao, Y.-J. Liu, S.~Tong, and Y.-Y. Wang, ``Adaptive finite-time
  neural network control of nonlinear systems with multiple objective
  constraints and application to electromechanical system,'' \emph{IEEE
  Transactions on Neural Networks and Learning Systems}, vol.~32, no.~12, pp.
  5416--5426, 2021.

\bibitem{glotfelter2017nonsmooth}
P.~Glotfelter, J.~Cort{\'{e}}s, and M.~Egerstedt, ``Nonsmooth barrier functions
  with applications to multi-robot systems,'' \emph{IEEE Control Systems
  Letters}, vol.~1, no.~2, pp. 310--315, 2017.

\bibitem{Glotfelter2021}
------, ``A nonsmooth approach to controller synthesis for {Boolean}
  specifications,'' \emph{IEEE Transactions on Automatic Control}, vol.~66,
  no.~11, pp. 5160--5174, 2021.

\bibitem{wang2016multiobjective}
L.~Wang, A.~D. Ames, and M.~Egerstedt, ``Multi-objective compositions for
  collision-free connectivity maintenance in teams of mobile robots,'' in
  \emph{55th IEEE Conference on Decision and Control}, 2016, pp. 2659--2664.

\bibitem{mitchell2022adaptation}
M.~Black and D.~Panagou, ``Adaptation for validation of a consolidated control
  barrier function based control synthesis,'' \emph{arXiv preprint}, no.
  arXiv:2209.08170, 2022.

\bibitem{lindemann2019stl}
L.~Lindemann and D.~V. Dimarogonas, ``Control barrier functions for signal
  temporal logic tasks,'' \emph{IEEE Control Systems Letters}, vol.~3, no.~1,
  pp. 96--101, 2019.

\bibitem{lindemann2019conflicting}
------, ``Control barrier functions for multi-agent systems under conflicting
  local signal temporal logic tasks,'' \emph{IEEE Control Systems Letters},
  vol.~3, no.~3, pp. 757--762, 2019.

\bibitem{jankovic2018robust}
M.~Jankovic, ``Robust control barrier functions for constrained stabilization
  of nonlinear systems,'' \emph{Automatica}, vol.~96, pp. 359--367, 2018.

\bibitem{xu2015robustnessofCBFs}
X.~Xu, P.~Tabuada, J.~W. Grizzle, and A.~D. Ames, ``Robustness of control
  barrier functions for safety critical control,'' in \emph{IFAC Conference on
  Analysis and Design of Hybrid Systems}, vol.~48, no.~27, 2015, pp. 54--61.

\bibitem{Alan2023AV}
A.~Alan, A.~J. Taylor, C.~R. He, A.~D. Ames, and G.~Orosz, ``Control barrier
  functions and input-to-state safety with application to automated vehicles,''
  \emph{IEEE Transactions on Control Systems Technology}, vol.~31, no.~6, pp.
  2744--2759, 2023.

\bibitem{boyd2004convex}
S.~Boyd and L.~Vandenberghe, \emph{Convex optimization}.\hskip 1em plus 0.5em
  minus 0.4em\relax Cambridge University Press, 2004.

\end{thebibliography}

\end{document}